\documentclass{aims}
\usepackage{amsmath}
\usepackage{comment}
  \usepackage{paralist}
  \usepackage{graphics} 
  \usepackage{epsfig} 
\usepackage{graphicx}  \usepackage{epstopdf}
 \usepackage[colorlinks=true]{hyperref}
\hypersetup{urlcolor=blue, citecolor=red}

\def\currentvolume{X}
 \def\currentissue{X}
  \def\currentyear{200X}
   \def\currentmonth{XX}
    \def\ppages{X--XX}
     \def\DOI{10.3934/mfc.2020012}

\newtheorem{theorem}{Theorem}[section]

\newtheorem{lemma}[theorem]{Lemma}

\theoremstyle{definition}

\usepackage{algpseudocode}
\usepackage{algorithm}
\usepackage{ctable}
\usepackage{threeparttable}

\newcommand{\tr}{\operatorname{tr}}
\newcommand{\I}{\mathcal{I}}
\newcommand{\rank}{\mathrm{rank}}
\usepackage{comment}

\title[AIMS: Average Information Matrix Splitting] 
      {AIMS: Average Information Matrix Splitting}

\author[Shengxin Zhu]{}


 \email{Shengxin.Zhu@xjtlu.edu.cn}

\thanks{This research is supported by Foundation of LCP(6142A05180501), Jiangsu Science and Technology Basic Research Program (BK20171237), Key Program Special Fund of XJTLU (KSF-E-21, KSF-P-02), Research Development Fund of XJTLU (RDF-2017-02-23) and partially supported by NSFC (No.11771002, 11571047, 11671049, 11671051, 6162003, and 11871339).
}

\thanks{$^*$ Corresponding author: Shengxin Zhu}
\begin{document}

\maketitle
\centerline{\scshape Shengxin Zhu$^*$}
\medskip
{\footnotesize
\centerline{Laboratory for Intelligent Computing and Financial Technology}
 \centerline{Department of Mathematics, Xi'an Jiaotong-Liverpool University }
   \centerline{Suzhou, 215123, P.R.China}
} 

\medskip

\centerline{\scshape Tongxiang Gu and Xingping Liu}
\medskip
{\footnotesize
 \centerline{Laboratory of Computational Physics}
 \centerline{Institute of Applied Physics and Computational Mathematics}
  \centerline{Beijing 100088, P.R.China}
}

\bigskip


\begin{abstract}
For linear mixed models with co-variance matrices which are not linearly dependent on variance component parameters,
we prove that the average of the observed information and the Fisher information can be split into two parts. The essential part enjoys a simple and computational friendly formula, while the other part which involves a lot of computations is a random zero matrix and thus is negligible.
\end{abstract}


\section{Introduction}

Many statistical methods require an estimation of unknown (co-)variance parameter(s).
The estimation is usually obtained by maximizing a log-likelihood function. In principle, one requires the \textit{observed information matrix}---the
negative Hessian matrix of the log-likelihood---to obtain a
maximum likelihood estimator according to the Newton-Raphson method \cite{Jenn76}. The
expected value of the observed information matrix is usually referred to as the
\textit{Fisher information matrix} or simply the \textit{information matrix}. It keeps
the essential information about unknown parameters and enjoys a simper formula. Therefore
it is widely used in many applications \cite{EH78}. The resulting algorithms is called the \textit{Fisher-scoring algorithm}
which is widely used in informetrics \cite{L87} and now is standard procedure in computational statistics \cite[p.30]{FS}.

The Fisher scoring algorithm is a success in simplifying the approximation of the Hessian
matrix of the log-likelihood. Still, evaluating elements of the Fisher information matrix remains as one of bottlenecks in a log-likelihood maximization procedure, which prohibits the use of Fisher scoring algorithm for large data sets.
In particular, the high-throughput technologies in biological science and engineering
mean that the size of data sets and the corresponding statistical models have suddenly
increased by several orders of magnitude. Further simplification of computational procedure is quite needed for applications of large scale statistical models, such as genome-wide association studies, which involves many
thousands parameters to be estimated \cite{WZW13}. 

The aim of this short note is to provide a concise mathematical result: the average of the observed information and the Fisher information can be split into two parts; the essential part enjoys a simple formula and is easy to compute, while the other part which involves a lot of computations is a random zero matrix and thus is negligible.  Such a spitting and approximation provides significant reduction in computations. What we should mention is that
the average information idea has been proposed in \cite{john95} for (co)variance matrices which
linearly depend on variance parameters. It results in an efficient breeding algorithm in \cite{GTC95}, and followed by \cite{M97}. However, previous results assume that the variance-variance matrix should be linearly dependent on the underlying variance parameters.  Here we prove that similar results still be obtained even the variance-variance matrix is not linearly dependent on the underlying variance parameters.

\section{Preliminary}

Consider the following widely-used linear mixed models \cite{Chen19,Chen19b,Gao19,Zuo20,Wang20}
\begin{equation}
y=X\tau+Zu+\epsilon, \label{eq:LMM}
\end{equation}
where $y\in \mathbb{R}^{n\times 1}$ is the observation, $\tau\in
\mathbb{R}^{p\times 1}$ is the vector of fixed effects, $X \in
\mathbb{R}^{n\times p}$ is the design matrix which corresponds to
the fixed effects, $u \in \mathbb{R}^{b\times 1}$ is the vector of
unobserved random effects, $Z \in \mathbb{R}^{n\times b}$ is the
design matrix which corresponds to the random effects. $\epsilon \in
\mathbb{R}^{n\times 1}$ is the vector of residual errors. The random effects, $u$, and the residual errors, $\epsilon$,
follow multivariate normal distributions such that $E(u)=0$,
$E(\epsilon)=0$, $u\sim N(0, \sigma^2 G)$, $\epsilon \sim N(0,
\sigma^2 R)$ and
\begin{equation}
\text{var}\left[\begin{array}{c}
u\\
\epsilon
\end{array}\right]=\sigma^{2}\left[\begin{array}{cc}
G & 0\\
0 & R
\end{array}\right],
\end{equation}
where $G\in \mathbb{R}^{b\times b }$, $R \in \mathbb{R}^{n\times
n}$. Typically $G$ and $R$ are parameterized matrices with unknown parameters to
be estimated. Precisely, suppose that $G=G(\gamma)$, $R=R(\phi)$, and denote $\kappa=(\gamma, \phi)$, $\theta=(\sigma^2,\kappa)$.
Estimating our main concern, the variance parameters $\theta$,
requires a conceptually simple nonlinear iterative procedure:
one has to maximize a
residual log-likelihood function of the form 
\cite[p.252]{Searle06} 
\begin{equation}
\ell_R=\mathrm{const}-\frac{1}{2}\{ (n-\nu)\log\sigma^2 +\log
\det(H) +\log \det(X^TH^{-1}X) +\frac{y^TPy}{\sigma^2} \},
\label{eq:ellR}
\end{equation}
where $H=R(\phi)+ZG(\gamma)Z^T$, $\nu=\rank(X)$ and
$$P=H^{-1}-H^{-1}X(X^TH^{-1}X)^{-1}X^TH^{-1}.$$
Here we suppose $X$ is full rank, say, $\nu=p$.
The first derivatives of $\ell_R$ is referred to as the
\textit{scores} for the variance parameters $\theta:=(\sigma^2,
\kappa)^T$\cite[p.252]{Searle06}:
\begin{align}
s(\sigma^2)&=\frac{\partial \ell_R}{\partial \sigma^2}
=-\frac{1}{2}\left\{
\frac{n-\nu}{\sigma^2}-\frac{y^TPy}{\sigma^4}\right \}, \label{eq:ssigma}\\
s(\kappa_i)&=\frac{\partial \ell_R}{\partial
\kappa_i}=-\frac{1}{2}\left\{ \tr(P\frac{\partial H}{\partial
\kappa_i}) -\frac{1}{\sigma^2}y^TP \frac{\partial H}{\partial
\kappa_i}P y  \right\}.
\end{align}
where $\kappa=(\gamma^T, \phi^T)^T$.
We shall denote $$S(\theta)=(s(\sigma^2),s(\kappa_1), \ldots,
s(\kappa_m))^T.$$ The negative Hessian of the log-likelihood
function is referred to as the \textit{observed information matrix}. We
will denote the matrix as $\I_O$.
\begin{equation}
\I_O=  -
\begin{pmatrix}
 \frac{\partial^2 \ell_R}{\partial \sigma^2\partial \sigma^2} &
 \frac{\partial^2 \ell_R}{\partial \sigma^2 \partial \kappa_1} &
 \cdots &\frac{\partial^2 \ell_R}{\partial \sigma^2 \partial \kappa_m}
 &
 \\
\frac{\partial^2 \ell_R}{\partial \kappa_1\partial \sigma^2 }
&\frac{\partial^2 \ell_R}{\partial \kappa_1\partial \kappa_1} & \cdots
&
\frac{\partial^2 \ell_R}{\partial \kappa_1\partial \kappa_m} \\
\vdots

 & \vdots & \ddots & \vdots
 \\
\frac{\partial^2 \ell_R}{\partial \kappa_m\partial \sigma^2}

 &  \frac{\partial^2 \ell_R}{\partial \kappa_m\partial \kappa_1} &
 \cdots& \frac{\partial^2 \ell_R}{\partial \kappa_m\partial \kappa_m}
\end{pmatrix}.
\label{eq:IO}
\end{equation}

\begin{algorithm}[!t]
\caption{Newton-Raphson method to solve {$S(\theta)=0$.}}
\begin{algorithmic}[1]
\State {Give an initial guess of $\theta_0$} \For{ $k=0, 1, 2,
\cdots$ until convergence }
 \State{Solve $\I_O(\theta_k) \delta_k=S(\theta_k)$,}
 \State{$\theta_{k+1}=\theta_k+\delta_k$}
\EndFor
\end{algorithmic}
\label{alg:NR}
\end{algorithm}
Given an initial guess of the variance parameter $\theta_0$, a
standard approach to maximize $\ell_R$ or find the root of the
score equation $S(\theta)=0$ is the Newton-Raphson method
(Algorithm \ref{alg:NR}), which requires elements of the observed information
matrix.

In particular,
\begin{equation}
\I_O(\kappa_i,\kappa_j) = \frac{\tr(P\ddot{H}_{ij})- \tr(P\dot{H}_iP \dot{H}_j)}{2}
 +\frac{2y^TP\dot{H}_iP\dot{H}_j Py -y^T
 P\ddot{H}_{ij}Py}{2\sigma^2}, \label{eq:OKK}
\end{equation}
where $\dot{H}_i=\frac{\partial H}{\partial \kappa_i}$ and
$\ddot{H}_{ij}=\frac{\partial ^2 H}{\partial \kappa_i \kappa_j}$. Each element $\I_O(\kappa_i,\kappa_j)$
involves two computationally intensive trace terms, which prohibits the practical use of the (exact) Newton-Raphson method for large data sets.

In practice, the
\textit{Fisher information matrix}, $\I=E(\I_O)$, is preferred. The elements of the Fisher information
matrix have simper forms than these of the observed information
matrix for example 
\begin{equation}
\I(\kappa_i, \kappa_j) = \frac{\tr{P\dot{H}_iP\dot{H}_j}}{2}. \label{eq:fisher}
\end{equation}
The corresponding algorithm is referred to as
the \textit{Fisher scoring algorithm} \cite{L87}. Still the element
$\I(\kappa_i, \kappa_j)$ of the Fisher information matrix
involves the computationally expensive trace terms.

When the variance-variance matrix $H$ is linearly dependent on the variance parameter, say $\ddot{H}_{ij}=0$, researchers noticed that  the \textit{average information matrix}
\begin{equation}
\frac{\I(\kappa_i, \kappa_j)+\I_O(\kappa_i,\kappa_j)}{2}
=\frac{y^TPH_iPH_jPy}{2\sigma_2}
\end{equation}
enjoys a simpler and more computational friendly formula \cite{GTC95}\cite{john95}.
\begin{equation}
\I_A = \frac{y^TP\dot{H}_iP\dot{H}_jPy}{2\sigma^2}.
\end{equation}

For more general covariance matrices when $\ddot{H}_{ij}\neq 0$, we still have the following nice property by the classical matrix splitting \cite[p.94]{Varga99}. The average information matrix can be split into two parts as follows \cite{ZGL16a}:
\begin{equation}
\frac{\mathcal{I}(\kappa_i, \kappa_j)+\mathcal{I}_O(\kappa_i,\kappa_j)}{2}
=\underbrace{\frac{y^TP\dot{H}_iP\dot{H}_jPy}{2\sigma_2}}_{\mathcal{I}_A(\kappa_i,\kappa_j)} \\
+\underbrace{\frac{\tr(P\ddot{H}_{ij})-y^TP\ddot{H}_{ij}Py/\sigma^2
}{4}}_{\mathcal{I}_Z(\kappa_i,\kappa_j)}.
\end{equation}
Such a splitting enjoys a nice property which is presented as our main result.

\section{Main result}

\begin{theorem}
Let $\I_O$ and $\I$ be the observed information matrix and the Fisher information matrix for the residual
log-likelihood of linear mixed model respectively, then the average of the observed information matrix and the Fisher information matrix can be split as $\frac{\I_O+\I}{2}=\I_A+ I_Z$, such that the expectation of $\I_A$ is the Fisher information matrix and $E({\I}_{Z})=0$.  \label{cor1}
\end{theorem}

Such a splitting aims to remove computationally expensive and
negligible terms so that a Newton-like method is applicable for
large data which involves thousands of fixed and random effects.
It keeps the essential information in the observed information matrix.  In this sense,
$\I_{A}$ is a good approximation which is data-dependent (on $y$) to the data-independent
Fisher information.
An Quasi-Newton iterative procedure is obtained by replacing
$\I_O$ with $\I_A$ in Algorithm \ref{alg:NR}.

\section*{Proof of the main result}
\subsection{Basic Lemma}

\begin{lemma}
\label{lemma1}
Let $y\sim N(X\tau, \sigma^2H)$,
be a random variable
and $H$ is a
symmetric positive definite matrix, then
$$P=H^{-1}-H^{-1}X(X^TH^{-1}X)^{-1}X^TH^{-1}$$ is a weighted projection
matrix such that
\begin{enumerate}
  \item $PX=0$;
  \item $PHP=P$;
  \item $\tr(PH)=n-\nu$, where $\operatorname{rank}(X)=\nu$;
  \item $PE(yy^T)=\sigma^2PH$.
\end{enumerate}
\label{lem:P}
\end{lemma}
\begin{proof}
The first two terms can be verified by direct computation. Since
H is a positive definite matrix, there exists $H^{1/2}$ such that
\begin{equation*}
\tr(PH)  =\tr(H^{1/2}PH^{1/2})=\tr(I-\hat{X}(\hat{X}^T\hat{X})^{-1}
\hat{X})
 =n-\rank(\hat{X})=n-\nu.
\end{equation*}
where $\hat{X}=H^{-1/2}X$. The 4th item follows because
\begin{equation*}
P
E(yy^T)=P(\mathrm{var}(y)+X\tau
(X\tau)^T)
=\sigma^2PH+PX\tau(X\tau)^T=\sigma^2PH.
\end{equation*}
\end{proof}

\begin{lemma}
Let $H$ be a parametric matrix of $\kappa$, and $X$ be an constant matrix, then the partial derivative of
the projection matrix $$ P=H^{-1}-H^{-1}X(X^TH^{-1}X)^{-1}X^TH^{-1} $$
with respect to $\kappa_i$ is
$
\dot{P}_i=-P\dot{H}_iP,
$
where $\dot{P}_i=\frac{\partial P}{\partial \kappa_i}$ and $\dot{H}_i=\frac{\partial H}{\partial \kappa_i}.$
\label{lem:PD}
\end{lemma}
\begin{proof} See Lemma B.3 \cite{GD11}.
Using the derivatives of the inverse of a matrix
$$
\frac{\partial A^{-1}}{\partial \kappa_i}=-A^{-1} \frac{\partial A}{\partial \kappa_i}A^{-1}.
$$
we have
\begin{align*}
\dot{P}_i =& \frac{\partial}{\partial \kappa_i}(H^{-1}-H^{-1}X(X^TH^{-1}X)^{-1}X^TH^{-1}) \\
  = & -H^{-1}\dot{H}_iH^{-1}+H^{-1}\dot{H}_iH^{-1}X(X^TH^{-1}X)^{-1}X^TH^{-1} \\
& -H^{-1}X(X^TH^{-1}X)^{-1}X^TH^{-1}\dot{H}_i  H^{-1}X(X^TH^{-1}X)^{-1}X^TH^{-1} \\
& + H^{-1}X(X^TH^{-1}X)^{-1}X^TH^{-1}\dot{H}_iH^{-1} \\
=&-H^{-1}\dot{H}_iP+H^{-1}X(X^TH^{-1}X)^{-1}X^TH^{-1}\dot{H}_iP
=-P\dot{H}_iP.\\
\end{align*}
\end{proof}

\subsection{Formulae of the observed information matrix}
\begin{lemma}
The element of the observed information matrix for the residual
log-likelihood \eqref{eq:ellR} is given by
\begin{align}
\I_O(\sigma^2,\sigma^2) &=
\frac{y^TPy}{\sigma^6}-\frac{n-\nu}{2\sigma^4}, \label{eq:ISS} \\
\I_O(\sigma^2,\kappa_i) &=
\frac{1}{2\sigma^4}y^TP\dot{H}_iPy , \label{eq:ISK}\\
\I_O(\kappa_i,\kappa_j) & = \frac{1}{2}\left\{\tr(P\dot{H}_{ij})-
\tr(P\dot{H}_iP \dot{H}_j)
 \right\}
  +\frac{1}{2\sigma^2}\left\{ 2y^TP\dot{H}_iP\dot{H}_j Py -y^T
 P\ddot{H}_{ij}Py\right\}.
\label{eq:IKK}
\end{align}
where $\dot{H}_i=\frac{\partial H}{\partial \kappa_i}$,
$\ddot{H}_{ij}=\frac{\partial^2 H}{\partial \kappa_i \partial \kappa_j}$.
\end{lemma}
 \begin{proof} See Result 4 in \cite{GD11}.
 The result in \eqref{eq:ISS} is standard according to the definition. The result in
 \eqref{eq:ISK} follows from the result in Lemma \ref{lem:PD} if one uses the score in \eqref{eq:ssigma}.
 The first term in \eqref{eq:IKK} follows because
 \begin{align*}
 \frac{\partial \tr(P\dot{H}_i)}{\partial
 \kappa_j}&=tr(P\ddot{H}_{ij})+\tr(\dot{P}_j\dot{H}_i)=\tr(P\ddot{H}_{ij})-\tr(P\dot{H}_jP\dot{H}_i) \quad ( \dot{P}_j=-P\dot{H}_jP).
 \end{align*}
 The second term in \eqref{eq:IKK} follows because of
 using the result in Lemma \ref{lem:PD},
 we have
 \begin{equation}
 -\frac{\partial (P\dot{H}_iP)}{\partial
 \kappa_j}=P\dot{H}_jP\dot{H}_iP-P\ddot{H}_{ij}P+P\dot{H}_iP\dot{H}_jP.
 \end{equation}
 Further note that $\dot{H}_i$, $\dot{H}_j$ and $P$ are symmetric.
The second term in \eqref{eq:IKK} follows because of
 $$y^TP\dot{H}_iP\dot{H}_jPy=y^TP\dot{H}_jP\dot{H}_iPy. $$

 \end{proof}
\subsection{Formulae of the Fisher information matrix}

The \textit{Fisher information matrix}, $\I$, is the
expected value of the observed information matrix, $ \I=E(\I_O).$
The Fisher information matrix enjoys a simpler formula than the
observed information matrix and provides the essential information provided by the data, and thus it
is a natural approximation to the negative Jacobian matrix.

\begin{lemma}
The elements of the Fisher information matrix for the residual
log-likelihood function in \eqref{eq:ellR} are given by
\begin{align}
\I(\sigma^2,\sigma^2) &=E(\I_O(\sigma^2,\sigma^2))=\frac{\tr(PH)}{2\sigma^4}=\frac{n-\nu}{2\sigma^4}, \label{eq:FISS}\\
\I(\sigma^2, \kappa_i)
&=E(\I_O(\sigma^2,\kappa_i))=\frac{1}{2\sigma^2}
\tr(P\dot{H}_i),\label{eq:FISK} \\
\I(\kappa_i, \kappa_j)
&=E(\I_O(\kappa_i,\kappa_j))=\frac{1}{2}\tr(P\dot{H}_iP\dot{H}_j)
.\label{eq:FIKK}
\end{align}
\end{lemma}
\begin{proof} The formulas can be found in \cite{Searle06}.
 Here we supply an alternative proof.
First note that $PX=0$, and according to Lemma \ref{lemma1}
\begin{align}
PE(yy^T)&=P(\sigma^2H+X\tau (X\tau)^T)=\sigma^2PH.
\label{eq:PEyy}
\end{align}
Then
\begin{align}
E(y^TPy)&=E(\tr(Pyy^T))=\tr(PE(yy^T))
 =\sigma^2\tr(PH) =(n-\nu)\sigma^2. \label{eq:Eypy}
\end{align}
Therefore
\begin{equation}
E(\I_O(\sigma^2,\sigma^2))=\frac{E(y^TPy)}{\sigma^6}-\frac{n-\nu}{2\sigma^4}=\frac{n-\nu}{2\sigma^4}.
\end{equation}
Second, we  notice that $PHP=P$. Applying the procedure in
\eqref{eq:Eypy}, we have
\begin{align}
E(y^TP\dot{H}_iPy)&= \tr(P\dot{H}_iPE(yy^T))  =\sigma^2\tr(P\dot{H}_iPH)  \nonumber \\
                  &=\sigma^2\tr(PHP\dot{H}_i) =\sigma^2\tr(P\dot{H}_i),  \label{eq:yphpy}\\
E(y^TP\dot{H}_iP\dot{H}_jPy) &=
\sigma^2\tr(P\dot{H}_iP\dot{H}_jPH) \nonumber \\
&=\sigma^2\tr(PHP\dot{H}_iP\dot{H}_j)
 =\sigma^2\tr(P\dot{H}_iP\dot{H}_j), \label{eq:yphphpy}\\
E(y^TP\ddot{H}_{ij}Py)&=\sigma^2\tr(P\ddot{H}_{ij}PH)=\sigma^2\tr(P\ddot{H}_{ij}).
\label{eq:yphhpy}
\end{align}
Substitute \eqref{eq:yphpy} into \eqref{eq:ISK}, we obtain
\eqref{eq:FISK}.   Substitute \eqref{eq:yphphpy} and
\eqref{eq:yphhpy} to \eqref{eq:IKK}, we obtain \eqref{eq:FIKK}.
\end{proof}

Using the Fishing information matrix as an approximation to the negative Jacobian results in the widely-used
Fisher-scoring algorithm \cite{L87}.
\subsection{Proof of the main result}
\begin{proof}
Let
\begin{align}
\I_A(\sigma^2,\sigma^2)&=\frac{1}{2\sigma^6}y^TPy;  \label{eq:ASS}\\
\I_A(\sigma^2,\kappa_i)&=\frac{1}{2\sigma^4}y^TP\dot{H}_iPy;\label{eq:ASK}\\
\I_A(\kappa_i,\kappa_j)&=\frac{1}{2\sigma^2}y^TP\dot{H}_iP\dot{H}_jPy;
\label{eq:AKK}
\end{align}
then we have
\begin{align}
\I_Z(\sigma^2,\sigma^2)&=0, \\
\I_Z(\sigma^2,\kappa_i)&=\frac{tr(P\dot{H}_i)}{4\sigma^2}-\frac{y^TP\dot{H}_iPy}{4\sigma^4}, \\
\I_Z(\kappa_i,\kappa_j)& =\frac{\tr(P\ddot{H}_{ij})-y^TP\ddot{H}_{ij}Py/\sigma^2,
}{4}
\end{align}
Apply the result in \eqref{eq:Eypy}, we have
\begin{equation}
E(\I_A(\sigma^2,\sigma_2))=\frac{(n-\nu)}{2\sigma^4}=\I(\sigma^2,\sigma^2).
\end{equation}
Apply the result in \eqref{eq:yphpy}, we have
\begin{equation}
E(\I_A(\sigma^2,\kappa_i))=\frac{\tr(P\dot{H}_i)}{2\sigma^2} \text{ and }E(\I_Z(\sigma^2,\kappa_i))=0.
\end{equation}
Apply the result in \eqref{eq:yphphpy}, we have
\begin{equation}
E(\I_A(\kappa_i,\kappa_j))=\frac{\tr(P\dot{H}_iP\dot{H}_j)}{2}=\I(\kappa_i,\kappa_j)
\end{equation}
and $E(\I_Z(\kappa_i,\kappa_j))=0$.

\end{proof}

\section{Discussion}
The average information splitting is one of the key techniques to reduce computation in the maximum likelihood methods \cite{WZW13}, other techniques like sparse inversion (see the state-of-art of the sparse inversion algorithm \cite{ZW19}) should also be implemented to evaluate the score of the log-likelihood. More details can be found in the review report \cite{ZW18}.
Since  Fisher information matrix is preferred not only in finding the variance of an
estimator and in Bayesian inference \cite{MN04}, but also in analyzing the asymptotic behavior
of maximum likelihood estimates \cite{EF76,Z97,Z98}.  Besides the
traditional application fields like genetical theory of natural
selection and breeding \cite{F30},  many other fields including theoretical physics and information geometry also use
the Fisher information matrix theory \cite{HS14}\cite{JJK04}\cite{PRE11}\cite{V07}. Therefore the average information matrix splitting techniques also provides promise in these directions.

\section*{Acknowledgments}
We would like to thank the anonymous reviewers for some constructive feedback.

\providecommand{\href}[2]{#2}
\providecommand{\arxiv}[1]{\href{http://arxiv.org/abs/#1}{arXiv:#1}}
\providecommand{\url}[1]{\texttt{#1}}
\providecommand{\urlprefix}{URL }

\medskip
Received xxxx 20xx; revised xxxx 20xx.
\medskip

\end{document}